\documentclass[11pt]{article}

\usepackage{geometry}
\usepackage{authblk}
\usepackage[dvipsnames]{xcolor}
\usepackage[colorlinks,linkcolor=blue,citecolor=PineGreen]{hyperref}
\usepackage[latin1]{inputenc}
\usepackage{amsmath}
\usepackage{graphicx}
\usepackage{amssymb}
\usepackage{amsthm}
\usepackage{verbatim}
\usepackage{epsfig}
\usepackage[labelfont=bf]{caption}
\usepackage{enumerate}
\usepackage{subfig}
\usepackage{epstopdf}

\newtheorem{lemma}{Lemma}
\newtheorem{remark}{Remark}

\newtheorem{proposition}{Proposition}

\def\begequarr{\begin{eqnarray}}
\def\endequarr{\end{eqnarray}}
\def\begequarrs{\begin{eqnarray*}}
\def\endequarrs{\end{eqnarray*}}
\def\begarr{\begin{array}}
\def\endarr{\end{array}}
\def\begequ{\begin{equation}}
\def\endequ{\end{equation}}
\def\lab{\label}
\def\begdes{\begin{description}}
\def\enddes{\end{description}}
\def\begenu{\begin{enumerate}}
\def\begite{\begin{itemize}}
\def\endite{\end{itemize}}
\def\endenu{\end{enumerate}}

\def\lef[{\left[\begin{array}}
\def\rig]{\end{array}\right]}
\def\qed{\hfill$\Box \Box \Box$}
\def\begcen{\begin{center}}
\def\endcen{\end{center}}
\def\begrem{\begin{remark}\rm}
\def\endrem{\end{remark}}

\def\begmat#1{\begin{bmatrix}#1\end{bmatrix}}
\def\begali#1{\begin{align}{#1}\end{align}}
\def\begalis#1{\begin{align*}{#1}\end{align*}}

\def\calh{{\cal H}}

\def\calm{{\cal M}}

\def\cals{{\cal S}}

\def\call{{\cal L}}

\def\callinf{{\cal L}_\infty}
\def\L2e{{\cal L}_{2e}}

\def\rea{\mathbb{R}}

\def\dist{\mbox{dist}}

\def\col{\mbox{col}}
\def\hal{{1 \over 2}}
\def\et{\varepsilon_t}

\def\II{I\&I}

\def\IJRNLC{{\it Int. J. on Robust and Nonlinear Control}}
\def\TAC{{\it IEEE Trans. Automatic Control}}

\def\EJC{{\it European Journal of Control}}

\def\AUT{{\it Automatica}}

\def\CSM{{\it IEEE Control Systems Magazine}}

\usepackage{color}
\newcommand{\bow}[1]{{\color{blue} #1}}

\usepackage[prependcaption,colorinlistoftodos]{todonotes}

\begin{document}

\title{\huge Orbital Stabilization of Nonlinear Systems via the Immersion and Invariance Technique}
\author[1,6]{Romeo Ortega}
\author[2]{Bowen Yi}
\author[3]{Jose Guadalupe Romero}
\author[4,5]{Alessandro Astolfi}
\affil[1]{\small Laboratoire des Signaux et Syst\'emes, CNRS-CentraleSup\'elec, France}
\affil[2]{\small Australian Centre for Field Robotics, The University of Sydney, Australia}
\affil[3]{\small Departamento Acad\'emico de Sistemas Digitales, ITAM, Mexico}
\affil[4]{\small Department of Electrical and Electronic Engineering, Imperial College London, UK}
\affil[5]{\small Department of Civil Engineering and Computer Science Engineering, University of Rome ``Tor Vergata'', Italy}
\affil[6]{\small Department of Control Systems and Informatics, ITMO University, Russia}

\maketitle
%

\begin{abstract}
Immersion and Invariance is a technique for the design of stabilizing and adaptive controllers and state observers for nonlinear systems recently proposed in the literature. In all these applications the problem is translated into  stabilization of equilibrium points. Motivated by some modern applications we show in this paper that the technique can also be used to treat the problem of {\em orbital stabilization}, where the final objective is to generate periodic solutions that are orbitally attractive. The feasibility of our result is illustrated with some classical mechanical engineering examples.
\end{abstract}
%
\section{Introduction}
\lab{sec1}
%
To solve the problems of designing stabilizing and adaptive controllers and state observers for nonlinear systems a technique, called Immersion and Invariance (I\&I), was proposed in \cite{ASTORT, ASTetal}. The first step in \II is the definition of a target dynamics, which is a lower dimensional system that captures the desired behavior that is to be imposed to the closed-loop system. In the second step of the design an invariant manifold in the state space of the system, such that the restriction of the system dynamics to this manifold is precisely the target dynamics, is defined. The design is completed defining a control law that renders this manifold attractive. While the second step of the design involves the solution of a partial differential equation (PDE)---corresponding to the Francis-Byrnes-Isidori (FBI) equations \cite{ISI}---the third step is a stabilization problem where it is desired to drive to zero the rest of state, \emph{i.e.}, the off-the-manifold coordinates, while preserving bounded trajectories. As shown in \cite{WANetal}, this latter step can also be translated into a contraction problem.

In all the examples mentioned above one deals with the problem of stabilization of equilibrium points---the desired equilibrium for the system in the stabilization and adaptive control scenarios, or the zero equilibrium for the state estimation error in observer design. In some modern applications---for example, walking robots, DC-to-AC power converters, electric motors and oscillation mechanisms in biology---the final objective is to induce a \emph{periodic orbit}. The main objective of this paper is to show that the \II technique can also be applied to solve this new problem, that is, the generation of attractive periodic solutions. The only modification required is in the definition of the target dynamics that, instead of having an asymptotically stable equilibrium, should be chosen with attractive periodic orbits.

The problem of designing controllers to ensure orbital stabilization has been studied in the literature for various applications and with different approaches. For mechanical systems of co-dimension one, the virtual holonomic constraints (VHC) method has been studied in the last two decades \cite{MAGCONtac,SHIetaltac,WESetal}. As explained in Remark \ref{rem4}, this technique can be viewed as a particular case of the \II approach proposed here. Starting with the pioneering works of \cite{FRAetal,FRAPOG,SPO}, orbital stabilization via \emph{energy regulation} has been intensively studied, mainly for pendular systems, where the basic idea is to pump energy into the system to swing-up the pendulum. Such an idea is further elaborated in \cite{ASTetalauto} as the pumping-and-damping method for the stabilization of the up-right equilibrium of pendular systems, yielding an almost globally asymptotically stable equilibrium. See also \cite{ARCetaltac,DUISTRejc} for more general cases, and \cite{ANDetal} for an interesting connection with chaos theory. In \cite{STASEPtac} the construction of passive oscillators for Lur'e dynamical systems using ``sign-indefinite" feedback static mappings, which is clearly related with the pumping-and-damping method of \cite{ASTetalauto}, has been proposed. A unified treatment of many of these methods has recently been reported in \cite{YIetalscl,YIetal}.

The remainder of the paper is organized as follows. In Section \ref{sec2} we give the problem formulation and present our main result. Section \ref{sec3} presents some examples, including a simple linear time-invariant (LTI) system and two models of mechanical system widely studied in the literature, as well as a power electronics system. The paper is wrapped-up with concluding remarks in Section \ref{sec4}.\\

\noindent {\bf Notation.} $I_n$ is the $n \times n$ identity matrix.  For $x \in \rea^n$, we denote square of the Euclidean norm $|x|^2:=x^\top x$. All mappings are assumed smooth. Given a function $f:  \rea^n \to \rea$ we define the differential operator $\nabla f:=\left(\frac{\displaystyle \partial f }{\displaystyle \partial x}\right)^\top$. Given a set $\mathcal{A} \subset \rea^n$ and a vector $x\in\rea^n$, we denote ${\rm dist}(x,\mathcal{A}) := \inf_{y \in \mathcal{A}}|x-y|$.
%

\section{Problem Formulation and Main Result}
\label{sec2}

%
We are interested in the generation of attractive periodic solutions for the  system
\begequ
\label{sys}
\dot{x} = f(x) + g(x)u,
\endequ
with state $x(t)\in\rea^n$, input $u(t)\in\rea^m$, with  $m < n$, and $g(x)$ full rank. More precisely, we want to define a mapping $v:\rea^n \to \rea^m$ such that the closed-loop system
$$
\dot x=f(x)+g(x)v(x)=:F(x)
$$
has a periodic solution $X: \rea_+ \to \rea^n$ that is orbitally attractive [Definition 8.2]\cite{KHA}. That is $X$ is such that
\begalis{
\dot X(t) & =F(X(t)), \\
X(t) & =X(t+T),\; \forall t \geq 0,
}
and the set defined by the closed orbit
$$
\{x \in \rea^n\;|\; x=X(t),\;0 \leq t \leq T\},
$$
is attractive and invariant.

The main result of the paper is given in the proposition below.

\begin{proposition}
\label{pro1}\rm
Consider the system \eqref{sys}. Assume we can find mappings
\begalis{
\alpha: \rea^p \to \rea^p, \quad \pi: \rea^p \to \rea^n, \quad \phi: \rea^n \to \rea^{n-p}, \quad v :\rea^n\times\rea^{n-p} \to \rea^m
}
with $p < n$, such that the following assumptions hold.
\begin{itemize}
    \item[\bf{A1}] (Target oscillator) The dynamical system
    \begequ
    \label{tardyn}
     \dot{\xi} = \alpha(\xi)
    \endequ
has non-trivial, periodic solutions $\xi_\star (t)=\xi_\star (t+T),\;\forall t \geq 0$, which are parameterized by the initial conditions $\xi(0)$, with $\xi(t)\in\rea^p$.
    \item[\bf{A2}] (Immersion condition) For all $\xi$,
    \begequ
    \label{fbi}
     g^\perp(\pi(\xi)) \left[f(\pi(\xi)) - \nabla \pi^\top(\xi) \alpha(\xi)\right]=0,
    \endequ
    where $g^\perp:\rea^n \to \rea^{n-m}$ is a full-rank left-annihilator of $g(x)$.
    \item[\bf{A3}] (Implicit manifold) The following set identity holds
    \begequ
    \label{impman}
     \calm:=\{x\in\rea^n ~|~ \phi(x)=0\} =
     \{x\in\rea^n ~|~ x=\pi(\xi),\;\xi \in \rea^p\}.
    \endequ
    \item[\bf{A4}] (Attractivity and boundedness) All trajectories of the system
    \begequ
    \begin{aligned}
     \dot{z} & = \nabla \phi^\top(x) [f(x)+g(x)v(x,z)], \\
     \dot{x} & = f(x) + g(x)v(x,z),
    \end{aligned}
    \lab{auxsys}
    \endequ
    with the initial condition $z(0) =  \phi(x(0))$, $z(t)\in \rea^{n-p}$, and the constraint
\begequ
\lab{concon}
v(\pi(\xi),0) = c(\pi(\xi)),
\endequ
where
\begequ
\lab{c}
c(\pi(\xi)) := [g^\top(\pi(\xi))g(\pi(\xi))]^{-1}g^\top(\pi(\xi))\left\{\nabla \pi^\top(\xi) \alpha(\xi) - f(\pi(\xi))\right\},
\endequ
are bounded and satisfy
    \begequ
    \label{ztozer}
     \lim_{t\to \infty} z(t) =0.
    \endequ
\end{itemize}
Then the system
\begequ
\lab{cloloosys}
\dot{x} = f(x) + g(x)v(x,\phi(x))
\endequ
is such that the periodic solution $x_\star (t)=\pi(\xi_\star (t))$ is orbitally attractive.
\end{proposition}
\begin{proof}
From \eqref{auxsys} with $z(0) =  \phi(x(0))$ we have that $z(t)=\phi(x(t))$ for all $t\ge 0 $. Replacing in \eqref{cloloosys}, and invoking the boudnedness assumption in {\bf A4}  ensures $x(t) \in \call_\infty$. Furthermore, since $\lim_{t\to\infty}z(t) = 0$ we conclude that the set $\calm$ is attractive. Now, \eqref{tardyn}, \eqref{fbi} and \eqref{c} imply
$$
\dot x|_{x=\pi(\xi),u=c(\pi(\xi))}=\dot \pi,
$$
consequently the set $\calm$ is invariant. From {\bf A4} we have that
\begalis{
\lim_{t\to\infty} z(t) =0 \; & \Rightarrow\;  \lim_{t\to\infty}\dist\{x(t),\calm_\star \} = 0,
}
where we have defined the attractive set
$$
 \calm_\star :=\{x\in\rea^n ~|~ x=\pi(\xi),\;\xi \in \Omega\},
$$
with $\Omega:=\{\xi \in \rea^p | \xi (t)=\xi_\star (t), \; 0 \le t \le T\}$. The orbital attractivity property is therefore proved.
\end{proof}

\begrem
\lab{rem1}
It is important to underscore that the only modification introduced to the main stabilization result of I\&I, that is, [Theorem 2.1]\cite{ASTetal}, is in the definition of the target dynamics in {\bf A1}. Instead of having an asymptotically stable equilibrium, now it possesses orbitally attractive periodic orbits.
\endrem

\begrem
\lab{rem2}
Ideally, we would fix a desired periodic trajectory $x_\star (t)=x_\star (t+T)$ and then impose on the mapping $\pi$ the additional constraint that $\xi_\star (t)=\pi^{\tt I}(x_\star (t))$ for all $t\ge 0$, where $\pi^{\tt I}:\rea^n \to \rea^p$ is a left inverse of $\pi$, that is, $\pi^{\tt I}(\pi(\xi))=\xi$. But this is a daunting task---even when the desired trajectory is imposed only on some of the state coordinates. Instead, we select target dynamics that has some periodic orbits, and fix some of the components of the mapping $\pi(\cdot)$ to ensure that the coordinates of interest have the same periodic orbit. Notice also that Proposition \ref{pro1} does not claim that $x$ converges to a \emph{particular} periodic orbit $\pi(\xi_\star)$, but only to (a $\pi$-mapped) one of the family of periodic orbits of the target dynamics, as illustrated in Fig. \ref{fig:I&I}.
\endrem

\begin{figure}
  \centering
  \includegraphics[width=8cm]{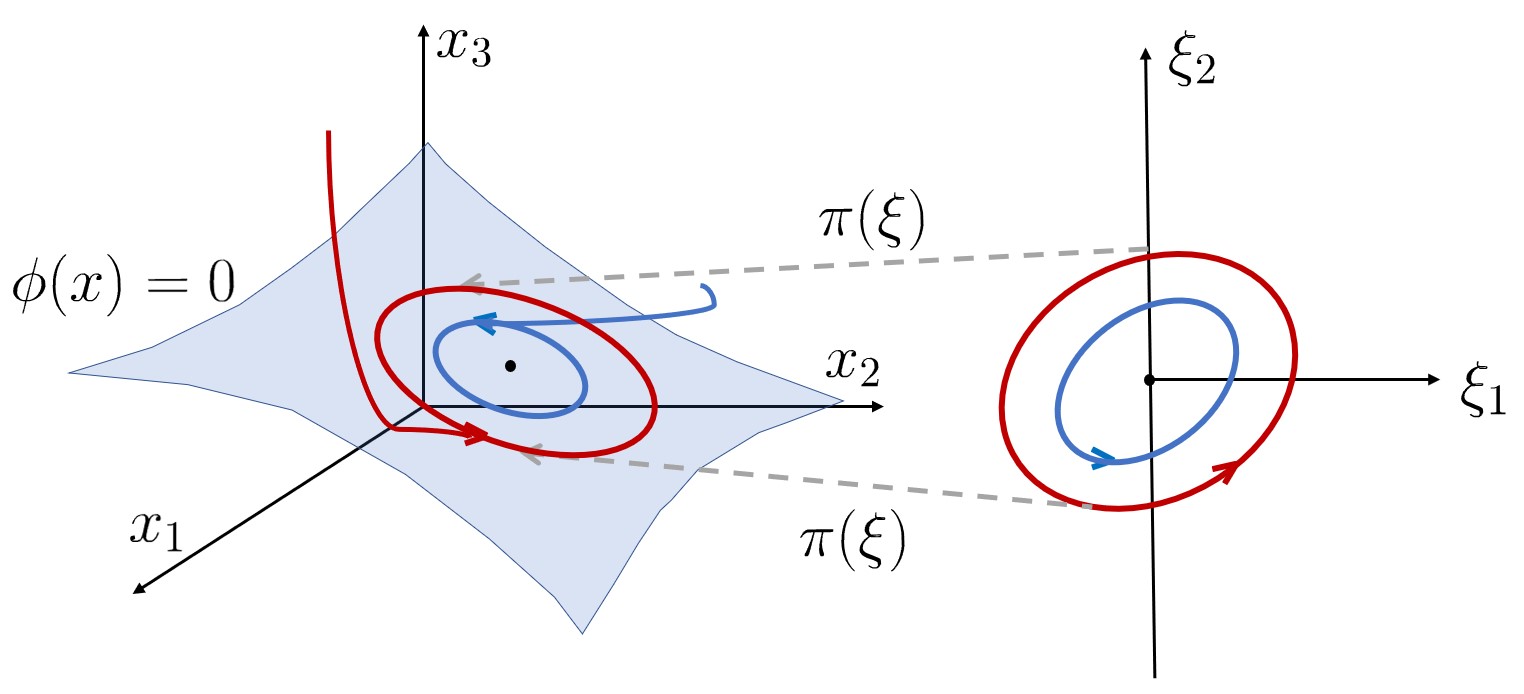}
  \caption{Schematic representation of Remark \ref{rem2}.}\label{fig:I&I}
\end{figure}

\begrem
\lab{rem3}
As indicated in \cite{WANetal}, the constraint condition \eqref{concon} is absent in [Theorem 2.1]\cite{ASTetal}. Also, to reduce the number of mappings to be found, we have expressed the FBI equation \eqref{fbi} projecting it into the null space of the input matrix $g(x)$. As shown in [Propositions 2 and 3]\cite{WANetal}, the stability condition {\bf A4} can be replaced by a contraction condition.
\endrem

\begrem
\lab{rem4}
The VHC method of \cite{MAGCONtac,SHIetaltac} is an alternative technique to induce periodic orbits, which can be viewed as a particular case of the \II design proposed here in the following sense. First of all, in contrast to our design that is applicable to arbitrary nonlinear systems of the form \eqref{sys}, the VHC method has been developed mainly for co-dimension one mechanical systems with $N$ degrees of freedom. However, see \cite{SHIetaltac10} for a recent extension. Second, in VHC the manifold to be rendered invariant, which is \emph{fixed a priori}, has the particular form
$$
\{(q,\dot q)\in \rea^N \times \rea^N\;|\; q_1=\psi_1(\xi),\;q_2=\psi_2(\xi),\dots,q_N=\xi,\;\xi \in \rea\},
$$
with $q$ the generalized coordinates. Therefore, the choice of target dynamics, which corresponds to the zero-dynamics of the system with output $q-\psi(q_1)$, is also restricted.\footnote{See  point 6 of [Section 2.1]\cite{ASTetal} for a discussion on the connection between zero-dynamics and \II.} Thirdly, with the notable exception of \cite{MOHetalaut}, attention has been centered only on rendering the manifold invariant, without addressing the issue of its attractivity, which is the main source of difficulty in \II.
\endrem
%
\section{Examples}
\label{sec3}
%
In this section we present four examples of application of Proposition \ref{pro1}. To illustrate the design procedure, we work out first a rather trivial LTI example. Then, we discuss the orbital stabilization problem for two models of mechanical systems, which have been widely studied in the control literature. Finally, a power electronics example is presented.
\subsection{LTI mechanical system}
\label{subsec31}
%
Consider the LTI system
\begalis{
\dot x_a & = x_b\\
\dot x_b &= - P x_a - R x_b + u,
}
with $x_a(t) \in \rea^2,~x_b(t)\in\rea^2,~u(t)\in \rea^2$, $R\in \rea^{2 \times 2}$ and $P \in \rea^{2 \times 2}$. The control objective is to induce an oscillation of unitary period to the component $x_a$ of the state. Towards this end, we follow step-by-step the procedure proposed in Proposition \ref{pro1}.

For Assumption {\bf A1} we pick $p=2$ and define the target dynamics as the linear oscillator $\dot \xi = J \xi$, where $J:=\begmat{0 &  1 \\ -1 & 0}$. Clearly,
$$
\xi(t)=e^{Jt}\xi(0)=\begmat{\cos t & -\sin t \\ \sin t & \cos t}\xi(0).
$$

It is easy to verify that the FBI equation \eqref{fbi} in Assumption {\bf A2} is satisfied selecting
\begalis{
\pi(\xi) &=T\xi, \;T:=\begmat{I_2 \\ J } \\ c(\pi(\xi)) & =K\pi(\xi),\;\;K:=\begmat{P & R+J }.
}
Also, it is clear that the condition \eqref{impman} in Assumption {\bf A3} holds selecting the mapping
\begalis{
\phi(x) &=x_b - Jx_a.
}
Finally,  Assumption {\bf A4} holds choosing
\begalis{
v(x,z) &=P x_a +(R+J)x_b-z,
}
which satisfies the boundary constraint \eqref{concon} and yields
\begalis{
\dot z & = -  z \\
\dot x_a &= x_b \\
\dot x_b &= J x_b - z.
}
Hence, $x \in \callinf$ and $\lim_{t\to\infty}z(t) = 0$ ensuring that $x$ converges to (a $\pi$-mapped) element of the family of periodic orbits of the target dynamics.

To verify the validity of the claim of the proposition, consider the control
$$
u=v(x,\phi(x))=(P-J)x_a+(R-J-I_2)x_b,
$$
yielding the closed-loop system $\dot x = A_{\tt cl}x$, with
$$
A_{\tt cl}:= \begmat{0 & I_2 \\ J & J - I_2},
$$
the eigenvalues of which are $\{i,-i,-1,-1\}$. The periodic function
$$
X (t):=\pi(\xi (t))=T \xi (t)=\begmat{I_2 \\ J}e^{Jt}\xi(0)
$$
satisfies $\dot X(t) = A_{\tt cl} X(t)$, hence it is a solution of the closed-loop system.
\subsection{Inertia Wheel Pendulum}
\lab{subsec32}
%
Our next example is the model of the inertia wheel pendulum (IWP) shown in Fig. \ref{fig1}. After a change of coordinates and a scaling of the input, the dynamic equations of the IWP are given by\footnote{All the details of the model can be found in \cite{ORTtac}.}
\begali{
\nonumber
    \dot x_1 &= x_3 \\ \nonumber
    \dot x_2 &=x_4\\ \nonumber
    \dot x_3 &= m \sin(x_1)  - bu \\
    \dot x_4 &= {u},
\lab{inewhe}
}
where $m>0,~b>0$ and $x(t) \in \cals \times \cals \times \rea \times \rea $, with $\cals$ the unit circle. The control objective is to lift the IWP from the \emph{hanging position} and to induce an oscillation of the link with a center at the upward position $x_1=0$.

\begin{figure}[htbp]
  \centering
 \includegraphics[width=4cm]{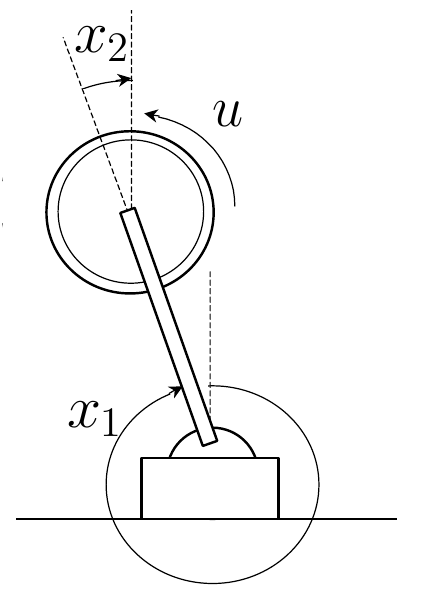}
  \caption{Inertia wheel pendulum}\label{fig1}
\end{figure}

\subsubsection{\II controller design}
\lab{subsubsec321}
%
We propose a simple undamped pendulum behavior for the target dynamics, \emph{i.e.}, $p=2$, and
\begali{
\nonumber
\dot \xi_1 & = \xi_2 \\
\nonumber
\dot \xi_2 & = - a\sin(\xi_1)
}
with $a$ a constant to be defined. The pendulum has a center at the downright equilibrium if $a>0$ or at the upright one if $a<0$. Consequently, it admits periodic orbits---defined by the level sets of the total energy function
$$
\calh_\xi(\xi):=\hal \xi_2^2 - a \cos(\xi_1),
$$
verifying Assumption  {\bf A1}. Now, motivated by the structure of \eqref{inewhe}, we propose the mapping
\begequ
\lab{piiwp}
\pi(\xi):=\begmat{\pi_1(\xi_1) \\ \pi_2(\xi_1) \\ \pi_1'(\xi_1)\xi_2 \\ \pi_2'(\xi_1)\xi_2 }
\endequ
with $\pi_i(\cdot),\;i=1,2,$ functions to be defined. We note that the first and the second components of the FBI equation \eqref{fbi} is satisfied by construction.

Consider the choice
$$
\pi_1(\xi)=\xi_1, \; \pi_2(\xi_1)=k \xi_1,
$$
with $k$ a constant to be defined. As a result, we get the linear mapping
\begequ
\lab{linmap}
\pi(\xi)=\begmat{1 & 0 \\ k & 0 \\ 0 & 1 \\ 0 & k}\xi=:T \xi.
\endequ
The implicit manifold description in Assumption {\bf A3} is satisfied selecting the linear mapping
\begequ
\lab{impman1}
\phi(x)=\begmat{ -k & 1 & 0 & 0  \\ 0 & 0 & - k & 1}x.
\endequ
After some simple calculations we see that the remaining two components of the FBI equation are solved, for any $k \neq -{1 \over b}$, with the choice
\begequ
\lab{ak}
a:={-m \over 1+bk},
\endequ
and the control
\begequ
\lab{ciwp}
c(\pi(\xi)) = -ak\sin(\xi_1).
\endequ
To complete our design it only remains to verify Assumption  {\bf A4} related to the auxiliary system \eqref{auxsys}. First, we compute the dynamics of the off-the-manifold coordinate $z=\phi(x)$ in closed-loop with the control $u=v(x,z)$ to get
\begalis{
\dot z_1 &= z_2\\
\dot z_2&=-km \sin(x_1) +(1+kb)v(x,z).
}
Let now
$$
v(x,z)= {1 \over 1+kb }\left[- \gamma_1 z_2 -\gamma_2 z_1 + km \sin(x_1)\right],\;\gamma_i>0,\;i=1,2,
$$
which, considering \eqref{ak} and \eqref{ciwp}, satisfies the constraint \eqref{concon}. It yields the closed loop dynamics
\begin{eqnarray*}
\dot z_1&=& z_2 \\
\dot z_2 &=& -\gamma_1 z_2 -\gamma_2 z_1 \\
\dot x_1&=& x_3 \\
\dot x_2&=&  x_4 \\
\dot x_3 &=& -a \sin(x_1)   + \et \\
\dot x_4&=&  {-ak \sin(x_1)+\et},
\end{eqnarray*}
where $\et$ are exponentially decaying terms stemming from the $z$-dynamics, which clearly verifies $\lim_{t\to\infty} z(t) = 0$ exponentially fast.  Now, since $x_1$ and $x_2$ live in the unit circle, and the control $v(x,z)$ is a function of $\sin(x_1)$, these two states are bounded. To complete the proof of boundedness of $x$, we recall the identity
$$
z=
\begin{bmatrix}
-kx_1 + x_2 \\ - kx_3 + x_4
\end{bmatrix},
$$
and consider the change of coordinates $x\mapsto (x_1,x_3,z_1,z_2)$, showing that we only need to check boundedness of $x_3$. Towards this end, we have the following lemma the proof of which, to enhance readability, is given in Appendix A.

\begin{lemma}
\label{lem1}\rm
Consider the nonlinear time-varying system
\begequ
\label{NLTV}
\begin{aligned}
    \dot{x}_1 & = x_3\\
    \dot{x}_3 & = - a\sin{( x_1)} + \et.
\end{aligned}
\endequ
with $(x_1,x_3) \in \mathcal{S}\times \rea$, where $\et$ satisfies
\begequ
\lab{bouepst}
|\et(t)| \le \ell_1 e^{-\ell_2 t},
\endequ
for some $\ell_1>0,~\ell_2 >0$. Then, $x_3(t)$ is bounded for all $t >0$.
\qed
\end{lemma}

Finally, as the unperturbed disk dynamics is given by the pendulum equation $\ddot x_1+a \sin(x_1)=0$, it  has a center at the upright equilibrium if $a>0$, or at the downright one if $a<0$. Note from Fig. \ref{fig1} that, unlike the classical pendulum equations, the upright equilibrium corresponds to $x_1=0$. Since the desired objective is to oscillate the link in the upper half plane we impose $a>0$, which translates into the constraint
\begequ
\lab{conk}
k<-{1 \over b},
\endequ
for $k$.

\begrem
\lab{rem5}
Lemma \ref{lem1} proves that the trajectories of an \emph{undamped} pendulum are bounded, in spite of the presence of an exponentially decaying term perturbing its velocity. In spite of the simplicity of the statement, and its obvious practical interest, we have not been able to find a proof of this fact in the literature. Hence, the result is of interest on its own.
\endrem

\subsubsection{Simulation results}
\lab{subsubsec322}

In this subsection we present some simulations for the IWP model in \eqref{inewhe}, with parameters $m=1.962$, $b=10$, in closed-loop with the proposed controller
$$
v(x,\phi(x))= {1 \over 1+kb }\left[- \gamma_1 (- kx_3 + x_4 )  -\gamma_2 (-kx_1 + x_2 ) + km \sin(x_1)\right],
$$
with $\gamma_1>0,~\gamma_2>0$ and $k$ verifying the constraint \eqref{conk}, which ensures  that the link oscillations are in the upper half plane. We concentrate our attention on  the link, since the disk has a similar behavior.

In Fig. \ref{fig2} we show a plot of $x_1$ vs $x_3$ for $a=0.1308$, (that is, $k=-1.6$), starting with the \emph{link hanging}, at $x(0)=[\pi, {1\over 3}\pi, 0, 0]$, and lifting it to oscillate in the upper half plane. Then, we illustrate the effect of the parameter $k$. In Fig. \ref{fig3} we show the transient behavior of $x_1$ and $x_3$ for $k \in \{-1.4, -1.6, -1.8, -2.0\} $  and the initial condition $x(0)=[{3\over4}\pi, {1\over 3}\pi, 0]$. Third, the effect of the initial conditions is illustrated in Fig. \ref{fig4}, where we have used the following values for the link position $x_1(0)\in \{{1\over 6}\pi, {1\over3}\pi, {2\over3}\pi, {5\over6}\pi\}$ and retained $x_{2}(0)={1\over3}\pi$, $x_{3}(0)=x_4(0)=0$, with the same value of $a=0.1308$. As expected from the analysis of the pendulum dynamics the link oscillates with an amplitude determined by the initial conditions and a frequency increasing when the magnitude of $a$ increases (that is, as $k$ decreases). Finally, to evaluate the effect of the gains $\gamma_1$ and $\gamma_2$ we carry out a simulation with the same initial conditions and gain $k$, but placing the poles of the off-the-manifold coordinate dynamics of the roots of the polynomial
$$
s^2+\gamma_1 s + \gamma_2=(s+p)^2,
$$
 with $p\in \{0.5, 1.0, 2.0, 3.0, 4.0\}$. As shown in Fig. \ref{fig5}, the transient degrades for slower rates of convergence of the off-the-manifold dynamics---as expected.

An animation of the system behavior may be found at \href{https://www.youtube.com/watch?v=Q5W9Kx0QbFo\&t=9s}{\texttt{\bow{www.youtube.com/watch?v=Q5W9Kx0QbFo\&t=9s}}}.

\begin{figure}[htp]
 \center
\includegraphics[width=.5\linewidth]{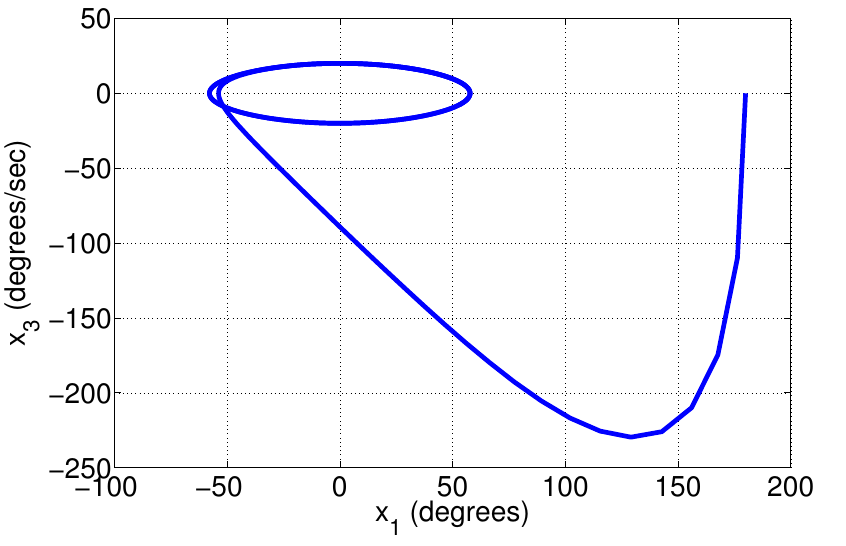}
  \caption{(IWP.) Plot of $x_1$ vs $x_3$ starting with the link hanging and lifting it to oscillate in the upper half plane. }
 \label{fig2}
\end{figure}

\begin{figure}[htp]
 \center
\includegraphics[width=.5\linewidth]{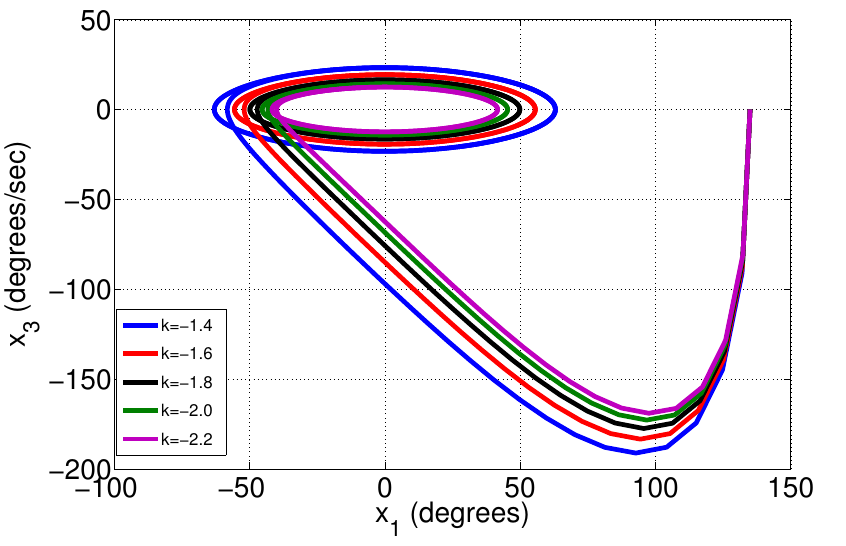}
   \caption{ (IWP.) Transient behavior of $x_1$ and $x_3$ with different gains $k$ and the same initial conditions.}
    \label{fig3}
\end{figure}

\begin{figure}[htp]
 \center
\includegraphics[width=.5\linewidth]{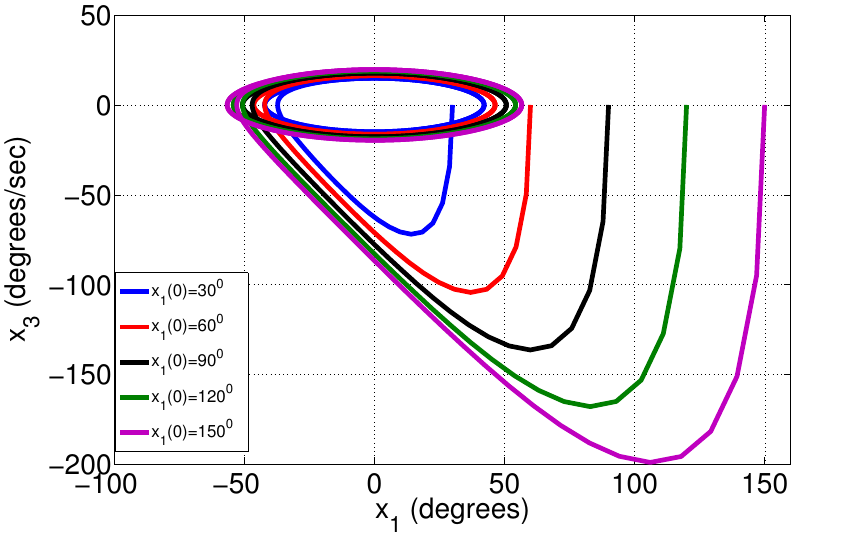}
  \caption{ (IWP.) Transient behavior of $x_1$ and $x_3$ with different initial conditions $x_1(0)$.}
    \label{fig4}
\end{figure}

\begin{figure}[htp]
 \center
\includegraphics[width=.5\linewidth]{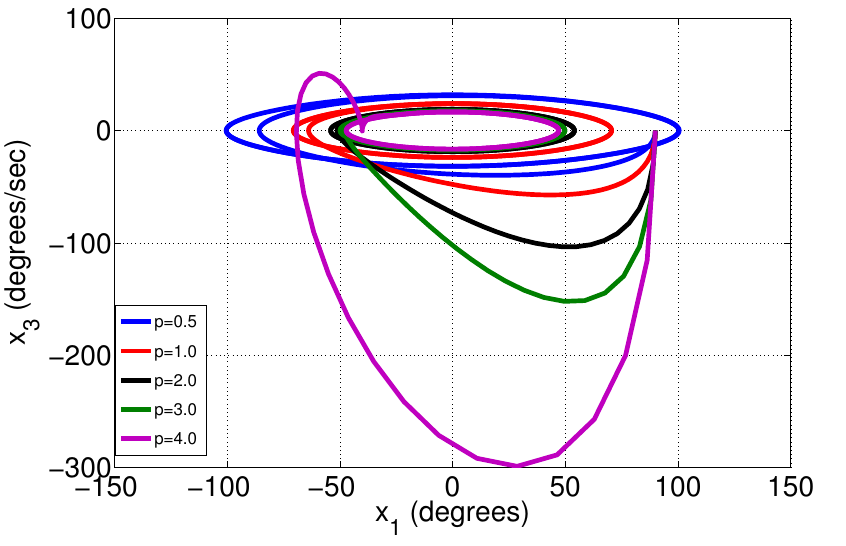}
  \caption{(IWP.) Transient behavior of the state $x_1$ and $x_3$ with different gains $\gamma_1$ and $\gamma_2$. }
    \label{fig5}
\end{figure}

\subsection{Cart-pendulum system}
\label{subsec33}
%
\begin{figure}[htbp]
  \centering
 \includegraphics[width=6cm]{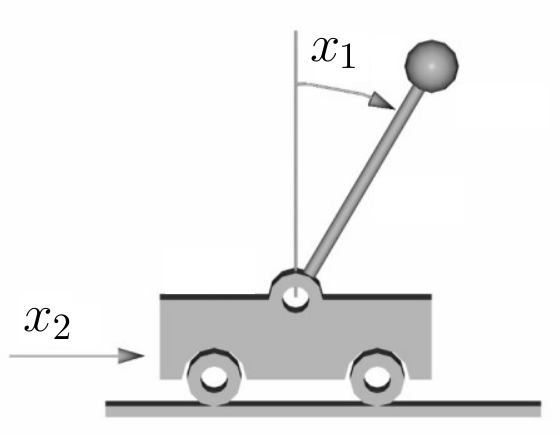}
  \caption{Pendulum on a cart system}
  \label{fig6}
\end{figure}

In this subsection we consider the model of a cart-pendulum system as depicted in Fig. \ref{fig6}. {After a partial feedback linearization and normalization of the dynamical model, we obtain the dynamics}\footnote{See \cite{ACOetal,TEE} for further detail.}
\begali{
\nonumber
    \dot x_1 &= x_3 \\ \nonumber
    \dot x_2 &=x_4\\ \nonumber
    \dot x_3 &= a_1 \sin(x_1)  -  a_2 \cos(x_1)u \\
    \dot x_4 &= u,
\lab{carpen}
}
where {$(x_1(t),x_2(t)) \in \mathcal{S}\times \rea$} are  the pendulum angle with the upright vertical and the cart position, respectively, $x_3(t)\in\rea,~x_4(t)\in \rea$ are their corresponding velocities, $u(t)\in \rea$ is the input, and $a_1>0$ and $a_2>0$ are physical parameters. The control objective is, starting with the link in the \emph{upper-half plane}, to induce an oscillation of the link with a center at the upward position $x_1=0$. Note that, for reasons to be explained below---unlike the IWP---we do not attempt to lift the pendulum from the hanging position.

\subsubsection{Controller design}
\lab{subsubsec331}
%
Similarly to the example in Subsection \ref{subsec32}, we select a two-dimensional target dynamics, \emph{i.e.}, $p=2$. In this case we consider a more general mechanical system of the form
\begali{
\nonumber
\dot \xi_1 & = \xi_2 \\
\lab{tardyn2}
\dot \xi_2 & = \alpha_2(\xi_1),
}
with $\alpha_2(\cdot)$ a function to be defined. The system has a total energy function
$$
\calh_\xi(\xi):=\hal \xi_2^2 +U(\xi_1),
$$
where
$$
U(\xi_1):=- \int_0^{\xi_1} \alpha_2(s)ds
$$
is its potential energy. Since the system is undamped, the derivative of its energy function is zero. Consequently, \emph{if} the potential energy has a minimum at zero, which is implied by the conditions
\begali{
\nonumber
\alpha_2(0) &=0\\
\lab{mincon}
\alpha'_2(0)& < 0,
}
then the target dynamics \eqref{tardyn2} admits periodic orbits---defined by the level sets of $\calh_\xi(\xi)$, and verifies Assumption  {\bf A1}.

We propose the mapping \eqref{piiwp}, with $\pi_i(\cdot),\;i=1,2,$ functions to be defined. From the third and the fourth components of the FBI equation \eqref{fbi} of Assumption  {\bf A2} we see that these functions must satisfy
\begequ
\lab{pdecarpen}
 a_1 \sin(\pi_1(\xi_1)) - a_2 \cos(\pi_1(\xi_1)) \bigg[ \pi_2''(\xi_1)\xi_2^2 + \pi_2'(\xi_1)\alpha_2(\xi_1) \bigg] = \pi_1''(\xi_1)\xi_2 + \pi_1'(\xi_1)\alpha_2(\xi_1).
\endequ
Factoring the elements depending on $\xi_2$ we conclude that $ \pi_1''(\xi_1)= \pi_2''(\xi_1)=0$, which implies that these functions should be \emph{linear}. Therefore, we select the mapping \eqref{linmap}, with $k$ a constant to be defined. The implicit manifold description of Assumption  {\bf A3} is satisfied with the linear mapping \eqref{impman1}.

Replacing the expressions of \eqref{linmap} in \eqref{pdecarpen} we obtain
\begali{
\alpha_2(\xi_1) & = {a_1 \sin(\xi_1) \over 1+ka_2 \cos(\xi_1)},
\nonumber
}
while the control must be chosen such that
\begequ
\lab{ccarpen}
c(\pi(\xi))={ka_1 \sin(\xi_1) \over 1+ka_2 \cos(\xi_1)}.
\endequ

To ensure that the potential energy has a minimum at zero we must verify the conditions \eqref{mincon}. Hence, we compute
$$
\alpha'_2(0)={a_1 \over 1+ka_2},
$$
and we must impose on $k$ the constraint
\begequ
\lab{stacon}
-{1 \over a_2} > k.
\endequ
With this choice, singularities are avoided in the interval $\cos(\xi_1) >  -{1 \over k a_2}$, which  contains the origin. Note that  the interval above is, unfortunately, strictly contained in the upper-half plane and controller singularities may appear during the transient---stymying the possibility to lift the pendulum for the lower-half plane and making local our stability result.

To complete our design it only remains to verify Assumption  {\bf A4} related to the auxiliary system \eqref{auxsys}. First, we compute the dynamics of the off-the-manifold coordinate $z=\phi(x)$ in closed-loop with the control $u=v(x,z)$ to get
\begalis{
\dot z_1 &= z_2\\
\dot z_2&=-ka_1 \sin(x_1) +[1+ka_2 \cos(x_1)]v(x,z).
}
Let the control law be
\begequ
\label{control_cartpen}
v(x,z)=  {1 \over 1+ka_2 \cos(x_1)}\left[- \gamma_1 z_2 -\gamma_2 z_1 + ka_1 \sin(x_1)\right],\;\gamma_i>0,\;i=1,2,
\endequ
which, considering \eqref{ccarpen}, satisfies the constraint \eqref{concon}.  It yields the closed loop dynamics
\begin{eqnarray*}
\dot z_1&=& z_2 \\
\dot z_2 &=& -\gamma_1 z_2 -\gamma_2 z_1 \\
\dot x_1&=& x_3 \\
\dot x_2&=&  x_4 \\
\dot x_3 &=&  {a_1 \sin(x_1) - a_2 \cos(x_1) \et \over 1+ka_2 \cos(x_1)}  \\
\dot x_4&=&  {a_1 \sin(x_1)+\et \over 1+ka_2 \cos(x_1)} ,
\end{eqnarray*}
which is such that $\lim_{t\to\infty} z(t) = 0$.  Now, since $x_1$ lives in the unit circle, and the control $v(x,z)$ is a function of $\sin(x_1)$ and $\cos(x_1)$, this state is bounded. Similarly to the inertia wheel pendulum example, we only need to verify boundedness of $x_3$. For, we have the following lemma, the proof of which is given in Appendix \ref{appb}.

\begin{lemma}
\label{lem2}\rm
Consider the nonlinear time-varying system
\begequ
\label{lem2sys1}
\begin{aligned}
\dot{w}_1 & = w_2\\
\dot{w}_2 & = {  { a_1 \sin(w_1) + \et} \over 1+ka_2 \cos(w_1)}
\end{aligned}
\endequ
with $(w_1(t),w_3(t)) \in \mathcal{S}\times \rea$, $a_1,a_2 >0$, $k$ verifying \eqref{stacon} and $\et$ satisfying \eqref{bouepst}. If the initial state satisfies
\begequ
\nonumber
{w_1(0)} \in (-\beta_\star, \beta_\star)
\endequ
with
$$
\beta_\star:= \arccos \left(-{1\over ka_2}\right),
$$
then, there exists $\ell_2^{\min} >0$ such that
$$
\ell_2 \ge \ell_2^{\min}
\quad \Longrightarrow \quad
w_1(t) \in (-\beta_\star,\beta_\star)\;\mbox{and}\;|w_3(t)| \leq M.
$$
\qed
\end{lemma}

To complete the proof we note that, with a suitable definition of $\et$, the right-hand side of $\dot x_3$ may be written in the form \eqref{lem2sys1} and observing that the exponential decay ratio of the $z$ dynamics---and consequently the parameter $\ell_2$---can be made arbitrarily large with a suitable selection of the gains $\gamma_1$ and $\gamma_2$.

\begin{remark} \rm
\lab{rem6}
As indicated in Lemma 2 stability of the closed-loop system is only established for large gains $\gamma_i>0$ $(i=1,2)$, that ensure a sufficiently fast convergence to the invariant manifold. Interestingly, although this requirement is imposed by the stability proof, we have not been able to observe instability in our simulations even for extremely small gains.
\end{remark}

\subsubsection{An alternative controller design}
\label{subsubsec332}

To enlarge the domain of attraction of the periodic orbit and remove the restriction of using high gains explained in Remark \ref{rem6},  we propose in this subsection an alternative controller design. For, we take the \emph{nonlinear} mapping
\begequ
\nonumber
\pi(\xi) =
\begmat{
\xi_1 \\ k(\xi_1)  \\ \xi_2 \\ k'(\xi_1)\xi_2
}.
\endequ
The implicit manifold description of Assumption  {\bf A3} is satisfied selecting the mapping
\begequ
\lab{impman2}
\phi(x)=\begmat{ x_2 - k(x_1) \\ x_4 - k'(x_1)x_3}.
\endequ
Some simple calculations prove that the FBI equations of Assumption  {\bf A2} are solved selecting the controller
\begequ
\label{a_pi_333}
c(\pi(\xi)) =
{ k''(\xi_1)\xi_2^2 + a_1 k'(\xi_1)\sin(\xi_1)
\over
1 + a_2 k'(\xi_1) \cos(\xi_1)
} ,
\endequ
together with the target dynamics
\begequ
\label{target333}
\begin{aligned}
 \dot{\xi}_1 & = \xi_2,\\
 \dot{\xi}_2 & = \rho(\xi_1) + \beta(\xi_1)\xi_2^2,
\end{aligned}
\endequ
where
\begalis{
 \rho(\xi_1) &:=
 { a_1\sin(\xi_1) \over 1 + a_2k'(\xi_1)\cos(\xi_1)}, \\
 \beta(\xi_1) & := -
 { a_2k''(\xi_1)\cos(\xi_1) \over 1 + a_2k'(\xi_1)\cos(\xi_1)}.
}
To enlarge the range of $x_1$ for which singularities are avoided we propose to select $k(\cdot)$, such that the denominator of the control \eqref{a_pi_333} is \emph{constant}, that is as the solution of the ordinary differential equation
\begin{equation}
\label{ode333}
1 + a_2 k'(s)\cos(s) = -a,
\end{equation}
with $a$ a constant to be defined. The solution of \eqref{ode333} is given by
\begequ
\label{k333}
k(s) = - {1+a \over a_2} \ln\bigg({1+\sin(s) \over \cos(s)}\bigg) + a_0,
\endequ
where we have added a constant $a_0$ that allows setting the center of the cart at any desired position. Notice that the function $k(\cdot)$ is well-defined in the interval $(-{\pi\over2},{\pi\over2})$. With this choice of $k(\xi_1)$, the functions $\rho(\xi_1)$ and $\beta(\xi_1)$ become
\begin{equation}
\label{alpbeta333}
\rho(\xi_1) =  -{a_1\over a}\sin(\xi_1), \quad \beta(\xi_1) = - {1+a\over a}\tan(\xi_1).
\end{equation}
Now, the target dynamics \eqref{target333}, is an undamped mechanical system with total energy function
\begin{equation}
\label{E333}
\calh_\xi(\xi) = {m(\xi_1)\over 2}\xi_2^2 + U(\xi_1),
\end{equation}
inertia
\begequ
\lab{mcarpen}
m(\xi_1) :=  \exp\left\{ {2} \int_0^{\xi_1} {(1+a)\over a} \tan(s) ds \right\}
=\Big|\cos(\xi_1) \Big|^{-2(1+{1\over a})}
\endequ
and potential energy
$$
U(\xi_1) := {a_1 \over a}\int_{0}^{\xi_1}  \sin(s) m(s) ds
= {a_1\over a+2} \cos(\xi_1)^{-(1+{2\over a})}
$$
for $\cos(\xi_1) >0$. From \eqref{mcarpen} we conclude that there {exist} constants $m_{\min}$ and $m_{\max}$ such that
$$
0 < m_{\min} \leq m(s) \leq m_{\max},\;\forall s \in \bigg(-{\pi\over2},{\pi\over2}\bigg).
$$
It is obvious that, with $a>0$, the potential energy $U(\xi_1)$ has a minimum at zero, ensuring Assumption  {\bf A1}.
To verify Assumption  {\bf A4} we define from \eqref{impman2} the off-the-manifold coordinates
\begin{equation}
\nonumber
\begin{aligned}
z_1 & = x_2 - k(x_1), \\
z_2 & = x_4 - k'(x_1)x_3,
\end{aligned}
\end{equation}
the dynamics of which are
$$
\begin{aligned}
 \dot{z}_1 & = z_2\\
 \dot{z}_2 & = \big[ 1+ a_2 k'(x_1)\cos(x_1)\big] u
 - \big[ k''(x_1)x_3^2 + a_1k'(x_1)\sin(x_1) \big]\\
 & = -a u - \big[ k''(x_1)x_3^2 + a_1k'(x_1)\sin(x_1) \big],
\end{aligned}
$$
where we have used \eqref{ode333} to get the second identity. We design the feedback law as
$$
v(x,z) = -
{1\over a} \bigg(k''(x_1)x_3^2 + a_1k'(x_1)\sin(x_1) - \gamma_1 z_1 - \gamma_2 z_2 \bigg),
$$
which satisfies \eqref{concon} and ensures $z(t) \to 0$ exponentially fast.

Similarly to the analysis of the previous subsection, we only need to prove the boundedness of the subsystem $x_1,x_3$ in closed-loop with the control given above, which is given by
\begequ
\label{timevarying333}
\begin{aligned}
 \dot{x}_1 & = x_3\\
 \dot{x}_3 & = \rho(x_1) + \beta(x_1)x_3^2 - {a_2\over a}\cos(x_1)(\gamma_1 z_1 + \gamma_2z_2).
\end{aligned}
\endequ
Computing the derivative of the energy function $\calh_\xi(x_1,x_3)$, defined in \eqref{E333}, along the trajectories of \eqref{timevarying333} we get that
\begalis{
\dot \calh_\xi & = -m(x_1)x_3 {a_2\over a}\cos(x_1)(\gamma_1z_1+\gamma_2z_2) \\
& \le  {m_{\max}a_2 (\gamma_1 + \gamma_2)\over a}|x_3||z(0)|\exp(-\ell_2t),\;x_1 \in \bigg(0,{\pi\over2}\bigg).
}
Now, from the fact that
$$
U(x_1) \ge U(0) =0,\quad x_1 \in \bigg(-{\pi \over 2}, {\pi \over 2}\bigg),
$$
we obtain the following inequality
$$
|x_3| \le \sqrt{{2\over m_{\min}} \calh_\xi(x)},
$$
from which we obtain the bound
\begequ
\lab{boudotene}
\dot \calh_\xi \le \ell_3 \sqrt{\calh_\xi(x)} \exp(-\ell_2t),
\endequ
where we have used the definition
$$
\ell_3:= {m_{\max}a_2|z(0)| (\gamma_1 + \gamma_2)\over a}\sqrt{2\over m_{\min}}.
$$
Finally, consider the auxiliary dynamics
$$
\dot{p}  = \ell_3 \sqrt{p}  \exp(-\ell_2 t),
$$
with $p(0)\ge0$, the solution of which is
$$
\sqrt{p(t)} =
{\ell_3\over 2\ell_2}(1- \exp(-\ell_2t)) + \sqrt{p(0)}.
$$
Clearly, $p(t)$ is bounded thus, applying the Comparison Lemma \cite{KHA} to \eqref{boudotene}, we conclude that $\calh_\xi(x(t))$, and consequently $x_1$ and $x_3$, are bounded.

\begin{remark} \rm
\lab{rem7}
The main advantage of the controller proposed in this subsection is that the pendulum can now move in the \emph{whole} upper-half plane. Another advantage is that stability is ensured for all gains $\gamma_1,\gamma_2>0$---this is in contrast with the controller of the previous subsection as indicated in Remark \ref{rem6}. Of course, the prize that is paid for these goodies is a significant increase in the controller complexity.
\end{remark}

\subsubsection{Simulation results}
\lab{subsubsec333}
%
In this subsection we first present some simulations of the cart-pendulum system \eqref{carpen} with $a_1=9.8$ and $a_2=1$, in closed-loop with the controller proposed in Subsection \ref{subsubsec331}, namely
$$
v(x,\phi(x))=
{1 \over 1+ka_2 \cos(x_1)}\left[- \gamma_1 (- kx_3 + x_4 )  -\gamma_2 (-kx_1 + x_2 ) + km \sin(x_1)\right],
$$
with $\gamma_1>0,~\gamma_2>0$ and $k$ verifying the constraint \eqref{stacon}.

In Fig. \ref{fig:cart-p1} we show a plot of $x_1$ vs $x_3$ for $k=-4$ and $\gamma_1=\gamma_2=2$, with initial conditions $x(0) = [{1\over5}\pi, 0, {1\over10}\pi,0]$. Note that a non-zero initial velocity is assumed for the link. The effect of the parameter $k$ is illustrated in Fig. \ref{fig:cart-p3}, with the values of $k \in \{-3,-4,-6\}$ and the same initial condition as before. {As shown in the figure, the parameter $k$ affects the period of the oscillation in a direct manner.} Fig. \ref{fig:cart-p4} illustrates the effect of the gains $\gamma_1$ and $\gamma_2$.

\begin{figure}[htp]
 \center
\includegraphics[width=.5\linewidth]{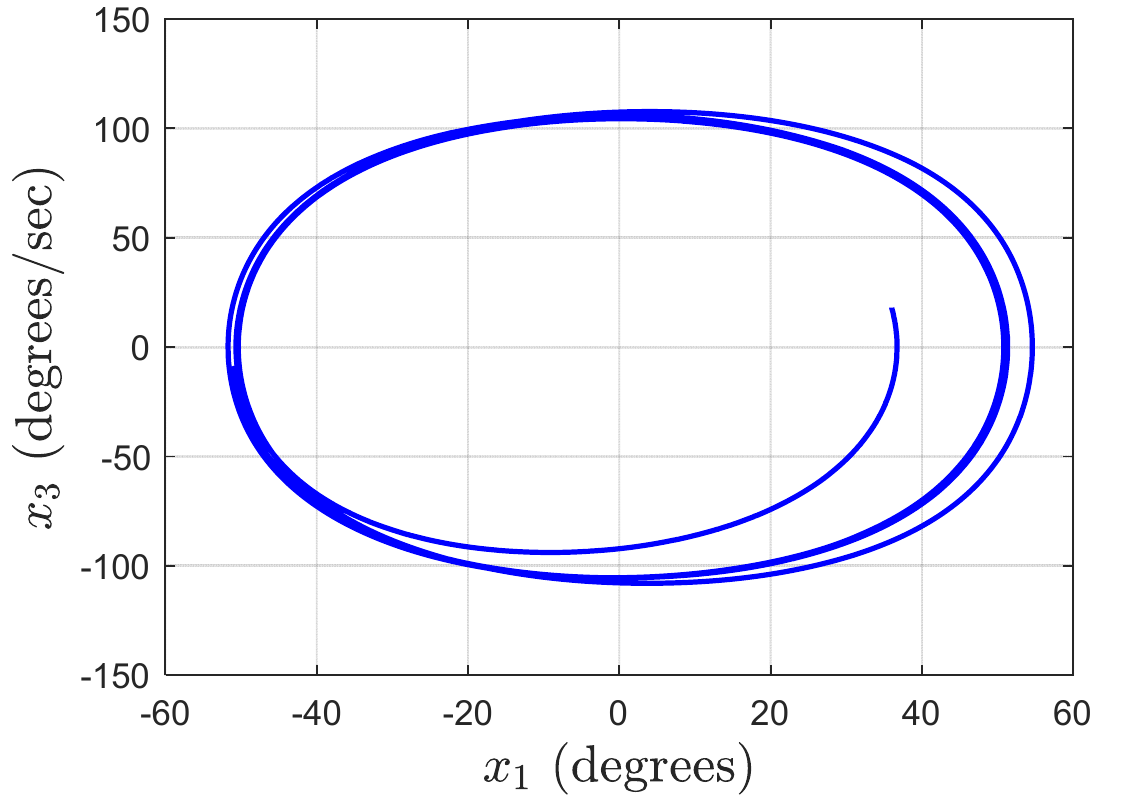}
  \caption{ (Cart pendulum.) Plot of $x_1$ vs $x_3$ starting with the link in the upper-half plane and a non-zero velocity.}
 \label{fig:cart-p1}
\end{figure}

\begin{figure}[htp]
 \center
\includegraphics[width=.5\linewidth]{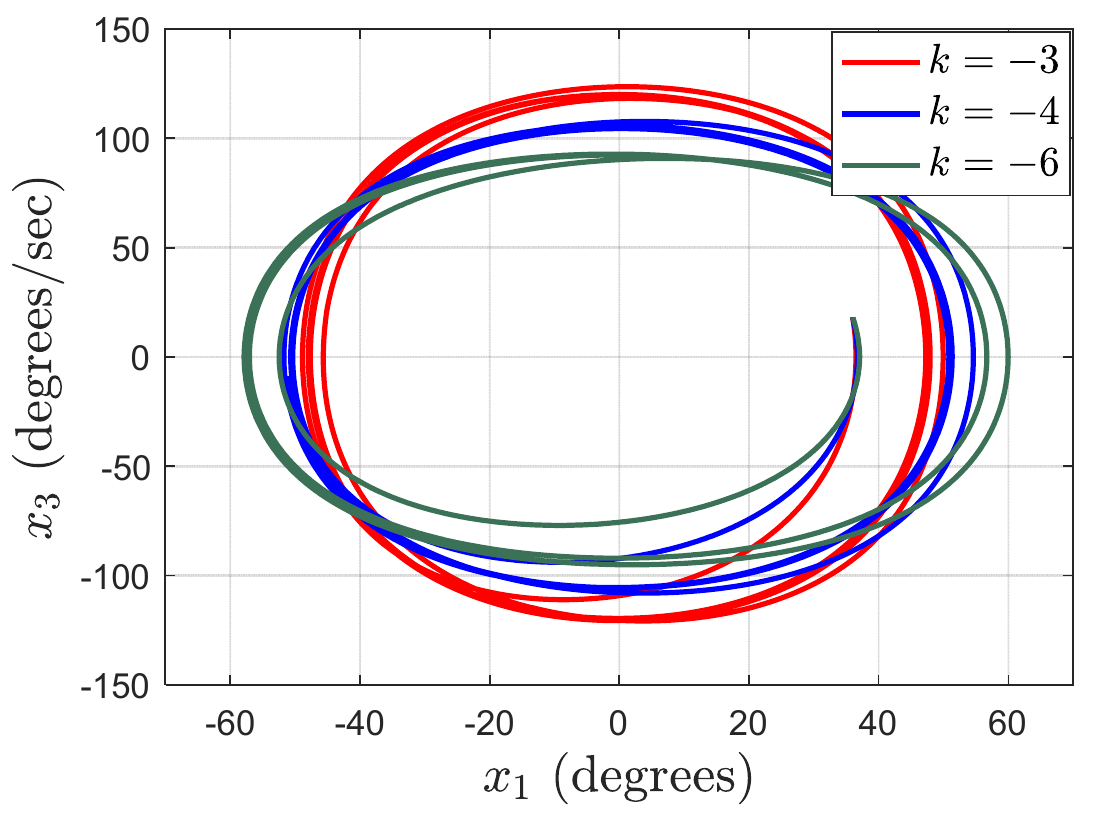}
  \caption{(Cart pendulum.) Transient behavior of the state $x_1$ and $x_3$ with different gains $k$ and the same initial conditions.}
 \label{fig:cart-p3}
\end{figure}

\begin{figure}[htp]
 \center
\includegraphics[width=.5\linewidth]{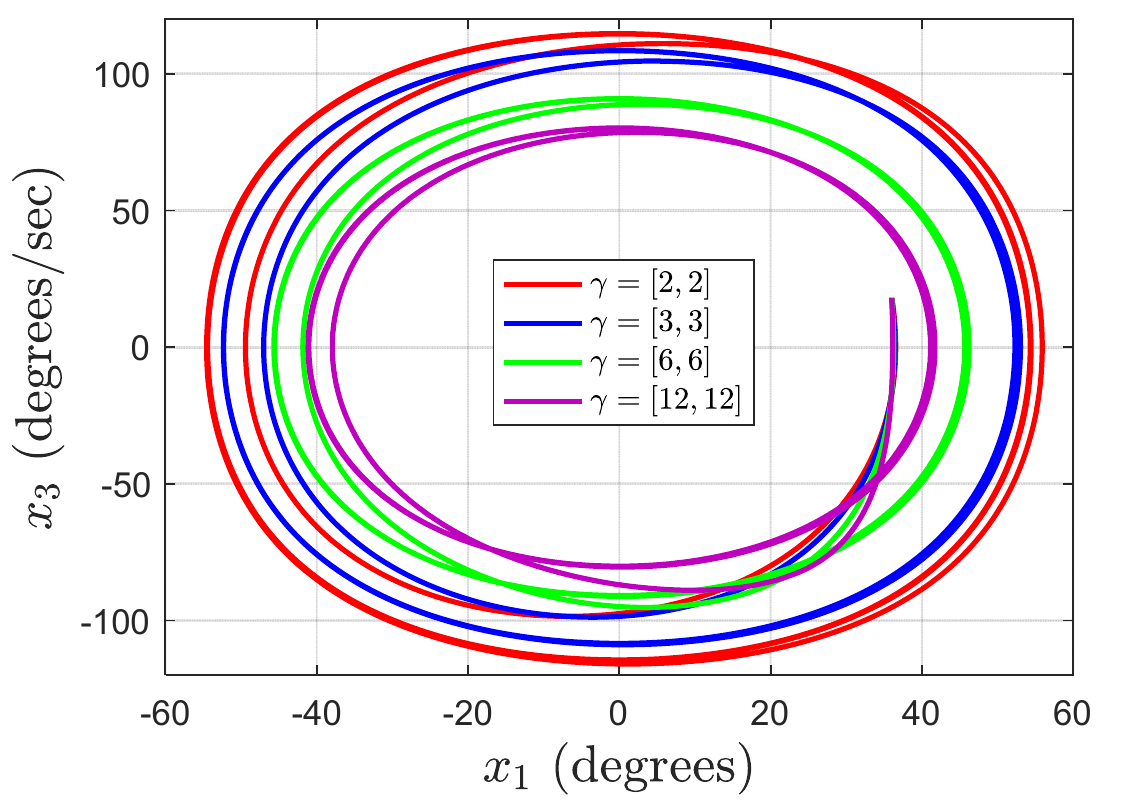}
  \caption{(Cart pendulum.) Transient behavior of the state $x_1$ and $x_3$ with different gains $\gamma_1$ and $\gamma_2$ and the same initial conditions.}
 \label{fig:cart-p4}
\end{figure}

{
We now give simulation results for the second design for the cart-pendulum system, that is, the controller
$$
v(x,\phi(x)) =
-
{1\over a} \bigg(k''(x_1)x_3^2 + a_1k'(x_1)\sin(x_1) - \gamma_2(x_2 - k(x_1)) - \gamma_1 (x_4 - k'(x_1)x_3) \bigg),
$$
with $k(x_1)$ given by \eqref{k333}. Fig. \ref{fig:cart-p5} displays the plot of $x_1$ vs $x_3$ for $a=2,a_0=0$ and $\gamma_1=\gamma_2 =1$, starting with the link \emph{closer to the horizontal} position and without any initial velocity, \emph{i.e.}, $x(0)=[{3\over10}\pi,-{1\over36}\pi,0,0]$. The effect of the parameter $a$ is illustrated in Fig. \ref{fig:cart-p6}.
}

 \begin{figure}[htp]
 \center
\includegraphics[width=.5\linewidth]{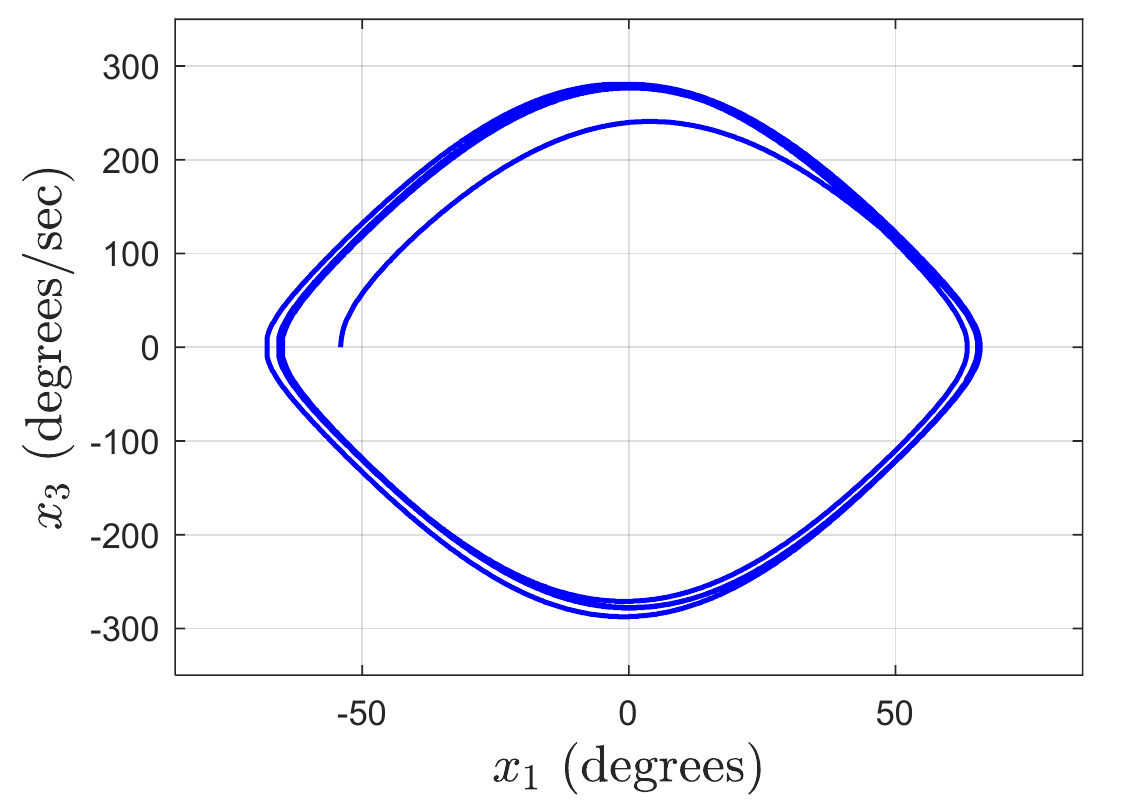}
  \caption{(Cart pendulum, the second controller.) Plot of $x_1$ vs $x_3$ starting with the link in the upper-half plane and zero velocity.}
 \label{fig:cart-p5}
\end{figure}

\begin{figure}[htp]
 \center
\includegraphics[width=.5\linewidth]{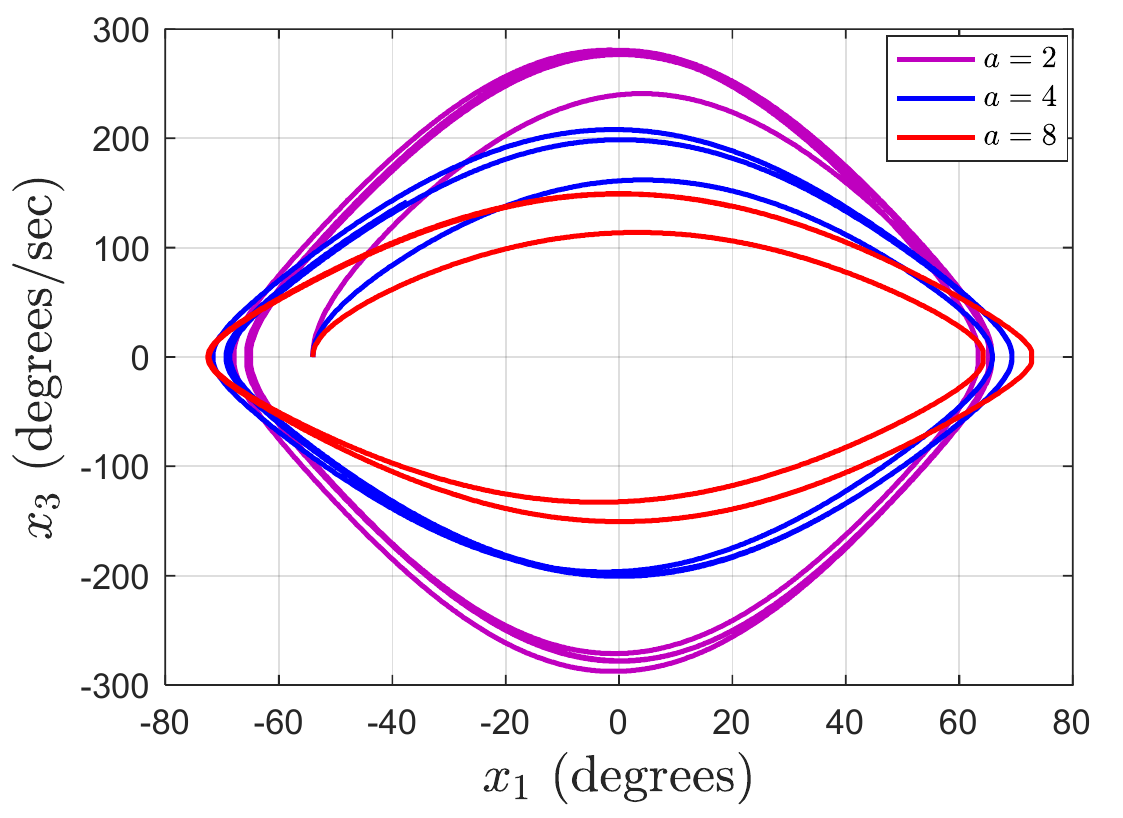}
  \caption{(Cart pendulum, the second controller.) Transient behavior of the state $x_1$ and $x_3$ with different gains $a$ and the same initial conditions. }
 \label{fig:cart-p6}
\end{figure}

An animation of the system behavior may be found at \href{https://www.youtube.com/watch?v=Q5W9Kx0QbFo\&t=9s}{\texttt{\bow{www.youtube.com/watch?v=Q5W9Kx0QbFo\&t=9s}}}.

\subsection{DC-AC Converter}
The last example is a three-phase DC-AC converter with a pure resistive load---see Fig. \ref{fig:converter}. The system dynamics can be described as \cite{ESCetal}
$$
\begin{aligned}
\dot{x}_1 & = - {1\over RC} x_1 + {1\over C}x_3 \\
\dot{x}_2 & = - {1\over RC} x_2 + {1\over C}x_4 \\
\dot{x}_3 & = - {1\over L} x_1 + {E\over L} u_1 \\
\dot{x}_4 & = - {1\over L} x_2 + {E\over L} u_2
\end{aligned}
$$
with the inductance $L>0$, the capacitance $C>0$, the capacitor voltages $(x_1,x_2)$, and the inductor currents $(x_3,x_4)$ in $\alpha\beta$ coordinates, where $E$ is the DC source voltage. The control objective is to generate a sinusoidal signal in the voltages $(x_1,x_2)$.

\begin{figure}[htp!]
  \centering
  \includegraphics[width=10cm]{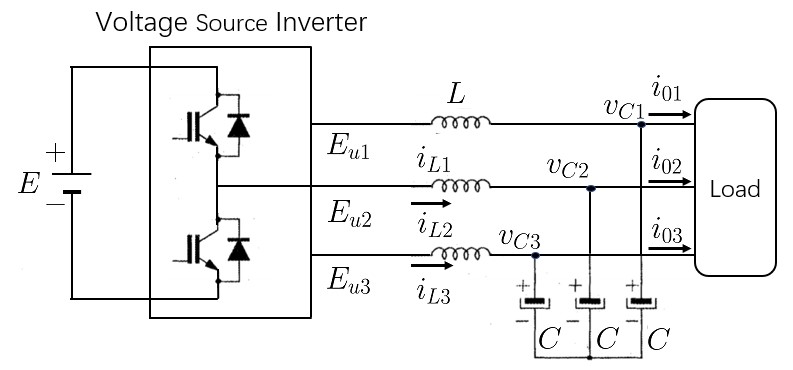}
  \caption{DC-AC converter}\label{fig:converter}
\end{figure}

We propose a target oscillator with $p=2$ described by the equation

\begin{equation}\label{eq:target_converter}
\dot{\xi} = \begmat{- (|\xi|^2 - A^2) & \omega \\ - \omega  & - (|\xi^2| - A^2)} \xi,
\end{equation}
with $A>0,~ \omega >0$. It should be underscored that the target system admits a \emph{unique attractive} orbit $\{\xi \in \rea^2 ||\xi| = A \}$ according to the analysis in \cite{YIetal}. Such a case is slightly different from {\bf A1} where \emph{a family of} orbits are parameterized by initial conditions.

Fixing the first two components of $\pi(x)$ as
$$
\pi_1(\xi) = \xi_1, \; \pi_2(\xi) = \xi_2
$$
and solving the FBI equation, we get
$$
\begin{aligned}
\pi_3(\xi) & = {1\over R}\xi_1 - C(|\xi|^2 - A^2) \xi_1 + C\omega \xi_2 \\
\pi_4(\xi) & = {1\over R}\xi_2 - C\omega \xi_1 - C(|\xi|^2 - A^2)\xi_2.
\end{aligned}
$$
The implicit manifold assumption {\bf A3} can be verified using
$$
\calm := \left\{ x\in \rea^4 ~\Bigg|~ \begmat{x_3 \\ x_4} - \beta(x_1,x_2) =0 \right\}
$$
with $\beta: \rea^2 \to \rea^2$ defined as $\beta(x_1,x_2) = \col(\pi_3(x_1,x_2), \pi_4(x_1,x_2))$. We then choose the off-the-manifold coordinate
$$
z := \begmat{x_3 \\ x_4}   - \beta(x_1,x_2),
$$
the dynamics of which is
$$
\dot{z} = {E\over L}v(x,z) - {1\over L}\begmat{x_1 \\ x_2} - \nabla \beta^\top(x_1,x_2) F(x),
$$
with the feedback law $v(x,z)$ to design and the mapping
$$
F_0(x) :=
\left[
\begin{array}{c}
- {1\over RC} x_1 + {1\over C}x_3 \\ - {1\over RC} x_2 + {1\over C}x_4
\end{array}
\right].
$$
Hence, we can construct the feedback law as
\begequ
\label{v:converter}
v(x,z) = {1\over E}\col(x_1,x_2) + {L \over E} \nabla \beta^\top(x_1,x_2) F_0(x) - \gamma z
\endequ
with a tunable parameter $\gamma >0$. The off-the-manifold coordinate $z$ has a globally exponentially stable equilibrium of zero, and the target oscillator is {almost globally exponentially orbitally} stable. It is relatively trivial to show the boundedness of the closed-loop dynamics with the aid of some basic perturbation analysis, unlike the case with the \emph{undamped} target oscillators, for instance, the examples of cart-pendulum and inertial wheel pendulum systems.

Finally, due to the boundedness of the state there always exist $E_*$ and $\gamma_*>0$ such that if the DC source voltage $E>E_*$ and the gain $\gamma \in (0, \gamma_*)$, we can guarantee $u(t) \in [-1, 1]$ for all $t>0$, thus satisfying the physical constraints. The obtained periodic signals, at the steady state, are
$$
x_1(t) = A\sin(\omega t + \phi), \;
x_2(t) = A\cos(\omega t + \phi),
$$
where the phase $\phi\in \rea$ depends on initial conditions.
\section{Concluding remarks}
\label{sec4}
%
We have shown that, by selecting the target dynamics in the well-known \II method \cite{ASTetal} to possess periodic orbits---instead of an asymptotically stable equilibrium---it is possible to solve the task of inducing orbitally attractive oscillations to general nonlinear systems. As usual with the  \II method, a large flexibility exists in the selection of the target dynamics and the definition of the manifold that is rendered attractive and invariant, which can be exploited to simplify the controller design. The result has been illustrated with some classical examples of mechanical and power electronics systems.t

Current research is under way to develop a systematic procedure to apply the technique that, at this stage, is used on a case-by-case basis. Towards this end, we plan to consider a ``more structured'' class of systems, for instance port-Hamiltonian systems, or a class of physically motivated systems like power converters and electric motors.

\appendix
%
\section{. Proof of Lemma 1}
\lab{appa}
%
Define the energy of the unperturbed pendulum \eqref{NLTV} as
$$
r(x) := {1\over 2} x_3^2 - a\cos (x_1),
$$
the time derivative of which along the trajectories of the system satisfies
\begequ
\label{dotv}
\dot{r} = x_3 \et.
\endequ
Note that
\begequ
\label{observation1}
r(x) \ge -a ,\quad \forall x~\in \rea^2,
\endequ
and
\begequ
\label{observation2}
|r(x(t))| \in \call_\infty
\quad \Rightarrow \quad
|x_3(t)| \in \call_\infty .
\endequ
Thus, we only need to prove that $r(x(t))$ is bounded.

From the bound
$$
|\dot{x}_3| \le |a\sin x_1| + |\et| \le a+ \ell_1,
$$
we have, for any $t\ge 0$
$$
\begin{aligned}
|x_3(t)| - |x_3(0)| & \le  \big|x_3(t) - x_3(0) \big|  \\
& =  \left|\int_{0}^{t} \dot{x}_3(s) ds \right| \\
& \le  \int_{0}^{t} \big|\dot{x}_3(s) \big| ds \\
& \le ( a+ \ell_1)t,
\end{aligned}
$$
thus
$$
|x_3(t)| \le |x_3(0)| + (a+ \ell_1)t.
$$

Recalling \eqref{dotv}, we have
$$
\begin{aligned}
\dot{r}  \le |x_3||\et|
     = \ell_3  e^{-\ell_2 t} + \ell_4 t  e^{-\ell_2 t},
\end{aligned}
$$
with $\ell_3:=\ell_1|x_3(0)|$ and $\ell_4:=\ell_1(a+\ell_1)$. Then,
$$
\begin{aligned}
r(x(t)) - r(x(0)) & \le \int_{0}^{t} (\ell_3  e^{-\ell_2 s} + \ell_4 s  e^{-\ell_2 s}) ds \\
& =
{\ell_3 \over \ell_2} \left( 1- e^{-\ell_2 t} \right)
+ {\ell_4 \over \ell_2^2}
\left(
1 - \ell_2 te^{-\ell_2t} - e^{-\ell_2t}
\right).
\end{aligned}
$$
As a result
$$
\lim_{t\to\infty} r(x(t)) \le r(x(0)) + {\ell_3 \over \ell_2} + {\ell_4 \over \ell_2^2},
$$
implying $r(x(t)) \in \call_\infty$ for all $t\ge 0$, which completes the proof.
%
\section{. Proof of Lemma 2}
\lab{appb}
%
Define the energy-like function
$$
\calh_w(w) := {1\over 2}w_2^2 + {a_1\over ka_2}\ln\big(|1+ka_2\cos(w_1)|\big).
$$
which is a first integral of the system \eqref{lem2sys1} in the absence of the decaying term.

This function is {lower bounded as
$$
\calh_w(w)
\ge
\calh_w^{\min} :=
{a_1\over ka_2} \ln
(-1-ka_2)
, \quad \text{for }w_1\in (-\beta_\star,\beta_\star),
$$}
We note that, in view of the constraint \eqref{stacon}, $-1-ka_2>0$, hence $\calh_w^{\min}$ is well-defined. Moreover,
$$
\lim_{|w_1| \to  \beta_\star} \calh_w(w) = +\infty.
$$
We also have the following bounds
\begequ
\label{ineq1}
    \big| w_2 \big| \le \sqrt{2(\calh_w(w)+ \calh_w^{\min})}
\endequ
and
\begequ
\label{ineq2}
    \big| 1 +ka_2\cos(w_1) \big| \ge \exp\bigg\{ {ka_2\over a_1} \calh_w(w) \bigg\}.
\endequ
Clearly,
$$
\begin{aligned}
    \dot{\calh}_w & = {w_2 \et \over 1+ ka_2 \cos(w_1)} \\
                & \le \left| {w_2  \over 1+ ka_2 \cos(w_1)} \right| \ell_1 \exp(-\ell_2 t) \\
               & \le \ell_1 \sqrt{2(\calh_w(w) + \calh_w^{\min})} \exp \left\{ -{ka_2\over a_1} \calh_w(w)\right\}  \exp(-\ell_2 t),
\end{aligned}
$$
where the last inequality has used \eqref{ineq1} and \eqref{ineq2}. To apply the Comparison Lemma we study the boundedness of the following one-dimensional auxiliary system
\begequ
\label{aux1}
\dot{r} =\ell_1 {\sqrt{2(r+\calh_w^{\min})}} \exp \left\{ -{ka_2\over a_1} r\right\}  \exp(-\ell_2 t).
\endequ
Note that $[-\calh_w^{\min},+\infty)$ is an invariant set for the differential equation \eqref{aux1}, with $r=-\calh_w^{\min}$ an equilibrium point. Therefore, we are interested only in the trajectories satisfying $r(t)+\calh_w^{\min} > 0$. In which case, $\Dot{r}>0$, and consequently $r(t)$ is a \emph{strictly increasing} function of time.

Define a function $F(r)$ as
$$
F(r) := \exp \left\{ -{ka_2\over a_1} r\right\}  - \sqrt{2(r+\calh_w^{\min})}.
$$
In view of the monotonicity, it is clear that there exists $r_0 >0$ such that
$$
F(r_0) = 0,
$$
and
$$
F(r) \ge 0 \quad \forall r> r_0.
$$
Therefore, there are two possible scenarios for system \eqref{aux1}:
\begin{enumerate}
  \item[1)] $r(t) < r_0$ for all $t>0$;
  \item[2)] there exists a time instant $t_1\ge 0$ such that $r(t) \ge r_0$ for all $t\ge t_1$.
\end{enumerate}

For Case 1), the boundedness of $r(t)$ follows immediately. For the second case, the dynamics \eqref{aux1} yields
$$
\begin{aligned}
\dot{r} & = \ell_1 \bigg[ \exp\bigg\{ -{ka_2 \over a_1}r \bigg\} - F(r) \bigg] \exp \bigg\{ -{ka_2 \over a_1}r \bigg\} \exp(-\ell_2t) \\
& \le  \ell_1  \exp \bigg\{ - 2{ka_2 \over a_1}r \bigg\} \exp(-\ell_2t),\quad
\forall t\ge t_1,
\end{aligned}
$$
where the first identity has used the definition of $F(r)$, and the second inequality has used the fact that $F(r(t))>0$ for $t\ge t_1$.

For the second case, by applying the Comparison Lemma we construct the auxiliary system
\begequ
\label{aux2}
\dot{v} = \ell_1 \exp \bigg\{ -2{ka_2 \over a_1}v \bigg\} \exp(-\ell_2t)
\endequ
with the initial condition
$
v(0) \ge r_0.
$

For the system \eqref{aux2}, we have
$$
 \dot{v}\exp\bigg\{ 2{ka_2\over a_1} v\bigg\} = \ell_1 \exp(-\ell_2t),
$$
thus integrating via variable separation we get
$$
\begin{aligned}
& \int_{0}^t \dot{v}(s)\exp\bigg\{ 2{ka_2\over a_1} v(s)\bigg\}ds = \int_{0}^t \ell_1 \exp(-\ell_2s)ds
\\
\Rightarrow\quad  &
\int_{v(0)}^{v(t)} \exp\bigg\{ k_0 v(s)\bigg\}dv = \int_{0}^t \ell_1 \exp(-\ell_2s)ds,
\end{aligned}
$$
with $k_0 := -2k{a_2\over a_1} >0$.

After some straightforward calculations, we then get
\begequ
\label{keyident}
\exp(-k_0v(t)) - \exp(-k_0v(0)) = k_0{\ell_1\over \ell_2}\bigg( \exp(-\ell_2t) - 1 \bigg).
\endequ
According to \eqref{keyident}, if
\begequ
\label{con_proof2}
\exp(-k_0v(0)) + k_0{\ell_1\over \ell_2}\bigg( \exp(-\ell_2t) - 1 \bigg) >0,
\endequ
we have
$$
v(t) = - {1\over k_0} \ln \Bigg[ \exp(-k_0v(0)) + k_0{\ell_1\over \ell_2}\bigg( \exp(-\ell_2t) - 1 \bigg) \Bigg].
$$
Using in \eqref{con_proof2} the following inequality
$$
- k_0{\ell_1\over \ell_2} < k_0{\ell_1\over \ell_2}\bigg( \exp(-\ell_2t) - 1 \bigg) ,
$$
we conclude that, if
\begequ
\label{ell_star}
\ell_2 > k_0 {\ell_1} \exp\big\{k_0 v(0)\big\} := \ell_0,
\endequ
the condition \eqref{con_proof2} holds for all $t>0$, implying that the solutions of \eqref{aux2} are bounded. Specifically,\footnote{We would like to point out that for the auxiliary system \eqref{aux2}, if $\ell_2 = \ell_0$ then
$$
\lim_{t\to\infty} v(t) = \infty;
$$
and for the case $\ell_2 \in (0,\ell_0)$ the system \eqref{aux2} has finite escape time.}
$$
\lim_{t\to\infty} |v(t)| = - {1\over k_0} \ln
\bigg\{
\exp\big(-k_0v(0)\big) - k_0{\ell_1\over\ell_2}
\bigg\}
< +\infty.
$$
Now, we return to the first auxiliary system \eqref{aux1} combining Case 1), if $\ell_2$ is large enough, we can obtain the boundedenss of $r(t)$ in terms of the Comparison Lemma and the boundedness of $v(t)$ for the auxiliary system \eqref{aux2}. Using the Comparison Lemma again and selecting
$$
\ell_2^{\min} :=
k_0 {\ell_1} \exp\bigg\{k_0 \max \big\{r_0, \calh_w^{\min} + \calh_w(w(0)) \big\}\bigg\},
$$
for $\ell_2 >\ell_2^{\min}$, the energy-like function $\calh_w(w(t))$ for the system \eqref{lem2sys1} is bounded for all $t>0$. Invoking the inequalities \eqref{ineq1} and \eqref{ineq2}, we complete the proof.

\section*{Acknowledgement}
The authors are grateful to two anonymous reviewers for their insightful remarks that helped to improve the quality of the paper. This paper is supported by Ministry of Science and Higher Education of the Russian Federation, project unique identifier RFMEFI57818X0271, and by the European Union's Horizon 2020 Research and Innovation Programme under grant agreement No 739551 (KIOS CoE).



\end{document}